\documentclass[10pt]{article}
\usepackage{amssymb}
\usepackage{mathtools}
\usepackage{amsthm}
\usepackage{microtype}

\newcommand{\Z}{\mathbb{Z}}

\newcommand{\R}{\mathbb{R}}

\newcommand{\SL}{\mathrm{SL}}
\newcommand{\SU}{\mathrm{SU}}
\newcommand{\GL}{\mathrm{GL}}
\DeclareMathOperator{\Aut}{Aut}
\DeclareMathOperator{\Sym}{Sym}

\DeclareMathOperator{\tr}{tr}

\newcommand{\Orth}{\mathrm{O}}

\renewcommand{\setminus}{\smallsetminus}
\newcommand{\norm}[1]{\left\lVert #1 \right\rVert}

\usepackage[colorlinks,citecolor=black,linkcolor=black,
urlcolor=black]{hyperref}

\newtheoremstyle{newdefinition}{}{}{}{}{\bfseries}{.}{ }{\thmname{#1}\thmnumber{ #2}\thmnote{ (#3)}}

\theoremstyle{newdefinition}
\newtheorem{theorem}{Theorem}
\newtheorem{lemma}[theorem]{Lemma}

\newtheorem{corollary}[theorem]{Corollary}

\begin{document}

\title{Maximizing the spectral gap for free bosonic CFTs}
\author{Hyeongjoe Chu, Yuchan Jung\\[4pt] Massachusetts Institute of Technology}
\date{}
\maketitle

\begin{abstract}
We study the spectral gap $\Delta_1$ of Narain CFTs, the smallest scaling dimension of a nonvacuum $U(1)^c_L\times U(1)^c_R$ primary, as a function on the Narain moduli space. Using the exceptional isomorphisms for $SO^+(2,2)$ and $SO^+(3,3)$, we prove the sharp bounds $\Delta_1\le 2/3$ at $c=2$ and $\Delta_1\le 3/4$ at $c=3$, with equality exactly at the $SU(3)_1$ and $SU(4)_1$ WZW models up to $O(c)\times O(c)$, and we describe local maxima of $\Delta_1$. For $4\le c\le 12$, we compare the Narain moduli space with the densest known Euclidean lattices: an arithmetic obstruction excludes five candidates, a computer-assisted classification shows that $E_8$, $K_{12}$, and $BW_{16}$ each admit a unique Narain structure up to lattice automorphisms, and a conceptual argument excludes the Leech lattice.
\end{abstract}

\section{Introduction}

A toroidal CFT of $c$ compact free bosons has left- and right-moving central charges $c_L=c_R=c$ and contains the corresponding abelian current algebras $U(1)^c_L\times U(1)^c_R$. We call such a theory a \emph{Narain CFT}. Such a theory can be completely encoded by a single algebraic object: an even self-dual lattice $\Lambda \subset \R^{c,c}$ of signature $(c,c)$, known as a \emph{Narain lattice}~\cite{MR834338}.

Throughout the paper, $\R^{c,c}=\R^c\oplus\R^c$ carries the split-signature bilinear form
\[
    J_c=\begin{pmatrix}I_c&0\\0&-I_c\end{pmatrix},
    \qquad
    \langle (x,y),(x',y')\rangle_{J_c}
    =x\cdot x'-y\cdot y'.
\]
Evenness and self-duality are understood with respect to this form. We also use the standard Euclidean inner product on the same vector space, with squared norm $\|(x,y)\|^2=|x|^2+|y|^2$. The two forms encode different CFT quantities: the current-algebra primary state associated with a lattice vector $(x,y)$ has scaling dimension and spin
\[
    \Delta=\frac{|x|^2+|y|^2}{2},
    \qquad
    s=\frac{|x|^2-|y|^2}{2}
      =\frac12\langle (x,y),(x,y)\rangle_{J_c}.
\]
In particular, evenness of the Narain lattice implies that $s\in\Z$.

An important invariant of a Narain lattice $\Lambda$ is its \emph{spectral gap}: the smallest scaling dimension among its nonvacuum current-algebra primary states. Thus
\[
    \Delta_1(\Lambda)
    =\frac12\min_{(x,y)\in\Lambda\setminus\{0\}}
      \bigl(|x|^2+|y|^2\bigr).
\]
It is therefore natural to ask which values $\Delta_1$ can attain, and in particular to identify the extremal Narain lattices: those maximizing $\Delta_1$. In Euclidean lattice terms, since self-duality forces every Narain lattice to have covolume one, this amounts to finding the densest lattice packing within the Narain class.

The level-one Wess--Zumino--Witten (WZW) model $\SU(r+1)_1$ corresponds to the Narain lattice
\[
    \Lambda(A_r)=\{(x,y)\in A_r^*\oplus A_r^*:x-y\in A_r\}
\]
\cite[Section~2.3]{MR4916198}, where $A_r\subset\R^r$ is a root lattice of type $A_r$ with roots of squared length $2$ and $A_r^*$ is its Euclidean dual. Replacing $A_r$ by $OA_r$ with $O\in O(r)$ changes $\Lambda(A_r)$ by $\operatorname{diag}(O,O)\in O(r)\times O(r)$, which preserves both $\Delta$ and the spin; since $O(c)\times O(c)$ acts separately on $x$ and $y$, Narain lattices related by it define the same CFT.

In this paper, we give a common algebraic approach to the sharp bounds at $c=2$ and $c=3$:

\medskip
\begin{theorem}\label{thm:c2}
    For every Narain lattice $\Lambda\subset\R^{2,2}$, the spectral gap satisfies $\Delta_1(\Lambda)\le 2/3$, with equality if and only if $\Lambda$ is equivalent under $O(2)\times O(2)$ to $\Lambda(A_2)$, the Narain lattice of the level-one WZW model $\SU(3)_1$.
\end{theorem}
\medskip
\begin{theorem}\label{thm:c3}
    For every Narain lattice $\Lambda\subset\R^{3,3}$, the spectral gap satisfies $\Delta_1(\Lambda)\le 3/4$, with equality if and only if $\Lambda$ is equivalent under $O(3)\times O(3)$ to $\Lambda(A_3)$, the Narain lattice of the level-one WZW model $\SU(4)_1$.
\end{theorem}

This extremal problem has been studied most directly through the modular bootstrap, in both its spinless and spinning forms. Hartman, Maz\'a\v{c}, and Rastelli~\cite{MR4075697} showed that, for free CFTs, the spinless modular bootstrap coincides with the Cohn--Elkies linear programming bound for sphere packing~\cite{MR1973059}. Afkhami-Jeddi, Cohn, Hartman, and Tajdini~\cite{MR4257711} subsequently applied the full spinning modular bootstrap, constructing linear functionals from derivatives at the self-dual point and optimizing them numerically via semidefinite programming. They obtained upper bounds on $\Delta_1$ for all $c\le 15$. The sharp bound at $c=1$ was proved analytically by constructing an explicit bootstrap functional~\cite[Section~3.2]{MR4257711}, while the $c=4$ bound follows from Viazovska's proof of the optimality of the $E_8$ lattice packing~\cite{MR3664816}. Their bounds at the other values of $c\le 15$ are numerical, and for $c\le 8$ they also constructed candidate Narain lattices by heuristic search~\cite[Section~3.3]{MR4257711}.

One important point is that the bootstrap method replaces the Narain CFT by a broader problem about positive modular-invariant spectra with $U(1)^c\times U(1)^c$ symmetry. This relaxation is well suited to producing general upper bounds, but it does not fully impose the arithmetic condition that the spectrum arises from an even self-dual lattice in $\R^{c,c}$. Indeed, apart from $c=1,2,4$, the numerical bounds of~\cite[Table~1]{MR4257711} differ significantly from the putative optimal values.

The cases $c=2$ and $c=3$ have attracted particular attention, with explicit constructions of candidate optimal lattices~\cite{MR4916198,MR4514857} and both numerical and analytic approaches to the sharp bounds. For $c=2$, the numerical spinning-bootstrap computation of~\cite[Eq.~(3.1)]{MR4257711}, truncated at derivative order $K=19$, gives the upper bound $\Delta_1\le 2/3+2\times10^{-11}$. Benjamin, Collier, Fitzpatrick, Maloney, and Perlmutter~\cite{MR4327142} subsequently obtained an exact formula for the scalar gap using harmonic analysis of the partition function~\cite[Eq.~(3.57)]{MR4327142} and showed that its maximum is $2/3$. Since every state of nonzero spin $s$ has $\Delta\ge|s|\ge1$, this also determines the full spectral gap $\Delta_1$. Our proof of Theorem~\ref{thm:c2} recovers this known result through an algebraic argument that also identifies all local maxima and parallels our proof at $c=3$. For $c=3$, the numerical spinning-bootstrap computation of~\cite[Table~1]{MR4257711} gives only the approximate bound $\Delta_1\lesssim0.8227$ without additional assumptions. Imposing integrality together with parity symmetry and a gap assumption on spin-one Virasoro primaries, Fitzpatrick and Li~\cite{MR4771416} improved the numerical upper bound to $\Delta_1\le0.750746$, close to the expected value $3/4$. The parity and gap assumptions are not consequences of the Narain lattice structure, so this bound does not cover the full Narain moduli space. Theorem~\ref{thm:c3} establishes the sharp bound for all Narain lattices at $c=3$, without these additional assumptions or a numerical approximation.

Our approach exploits the algebraic and arithmetic structure of Narain lattices directly. The common starting point for the two proofs is the pair of exceptional Lie group isomorphisms
\begin{gather*}
    \mathrm{SO}^+(2,2;\mathbb{R})\;\cong\;
\bigl(\mathrm{SL}_2(\mathbb{R})\times \mathrm{SL}_2(\mathbb{R})\bigr)/\{\pm(I,I)\},\\
\mathrm{SO}^+(3,3;\mathbb{R})\;\cong\;\mathrm{SL}_4(\mathbb{R})/\{\pm I\},
\end{gather*}
which we realize concretely through explicit linear isometries of the ambient quadratic spaces. For $c=2$, an explicit linear isometry $\Phi_2$ identifies $(\mathbb{R}^{2,2},J_2)$ with $M_2(\mathbb{R})$ equipped with the quadratic form $2\det$, carrying the reference Narain lattice to $M_2(\mathbb{Z})$. Every Narain lattice of central charge $2$ then becomes a lattice $g\,M_2(\mathbb{Z})\,h^{-1}$ with $g,h\in \mathrm{SL}_2(\mathbb{R})$. Its rank-one elements $(gu)(h^{-T}v)^T$ are built from the two unimodular planar lattices $g\mathbb{Z}^2$ and $h^{-T}\mathbb{Z}^2$, so the classical bound $\gamma_2 = 2/\sqrt{3}$ yields $\Delta_1\le 2/3$. The evenness of the Narain lattice then forces every minimal vector to be of rank one, so $\Delta_1$ is determined by these two planar lattices, with equality precisely when both are hexagonal; up to $O(2)\times O(2)$, the extremal Narain lattice is $\Lambda(A_2)$, which is the $\SU(3)_1$ WZW model. This splitting into two planar lattices reflects the well-known description of the $c=2$ Narain moduli space as a product of two modular curves, up to finite identifications, parametrized by the complex structure and K\"ahler moduli $(\tau,\rho)$~\cite[Section~2.4.2]{MR1288936}, and agrees with the exact scalar-gap formula in~\cite[Eq.~(3.57)]{MR4327142}.

For $c=3$, we proceed analogously through an explicit isometry $\Phi_3$ identifying $(\mathbb{R}^{3,3},J_3)$ with $\bigwedge^2\mathbb{R}^4$ equipped with the wedge pairing, under which the reference lattice becomes $\bigwedge^2\mathbb{Z}^4$ and $\mathrm{SO}^+(3,3;\mathbb{R})$ acts through $\mathrm{SL}_4(\mathbb{R})$. Thus, for each Narain lattice $\Lambda$, there is a matrix $g\in\mathrm{SL}_4(\mathbb{R})$ such that $\Phi_3(\Lambda)=\bigwedge^2L$, where $L=g\mathbb{Z}^4$ is a unimodular Euclidean lattice in $\mathbb{R}^4$. Evenness combines with the Pl\"ucker relation and the bound on the minimum squared length to force every minimal vector to be decomposable. The extremal problem therefore reduces to bounding the minimal squared covolume of a rank-two sublattice of $L$, i.e., to the Hermite--Rankin constant $\gamma_{4,2}=3/2$ proved by Rankin~\cite[Theorem~5]{MR55380}. This gives $\Delta_1\le 3/4$, with equality precisely when $L$ is similar to $D_4$, by Rankin's characterization of the equality case. Transporting $\bigwedge^2L$ back to $\R^{3,3}$ via $\Phi_3^{-1}$ then gives, up to $O(3)\times O(3)$, exactly the lattice $\Lambda(A_3)$, the Narain lattice of the $\SU(4)_1$ WZW model. In contrast with $c=2$, $\Delta_1$ also has local maxima that are not global, for instance the one coming from $A_4$. For related results on minimal vectors in the second exterior power of a lattice, see~\cite{MR1467163}.

In the second half of the paper we turn to higher central charge and ask a complementary question: for $4\le c\le 12$, can the putatively densest Euclidean lattices in dimension $2c$ arise as the Euclidean lattices underlying Narain lattices? We first record an integral form of the criterion of~\cite[Lemma~B.1]{MR4257711} for simultaneous Euclidean and Lorentzian Gram matrices: a determinant-one Gram matrix $E$ underlies a Narain lattice if and only if there exists an even symmetric integral matrix $K$ with $KE^{-1}K=E$ and $\operatorname{tr}(E^{-1}K)=0$. This yields an arithmetic obstruction that rules out the best known lattices in dimensions $10$, $14$, $18$, $20$, and $22$. For $E_8$, $K_{12}$, and $BW_{16}$ we show, by an exhaustive computer enumeration of the solutions $K$ modulo the automorphism group, that each admits a Narain structure that is unique up to automorphisms of the underlying Euclidean lattice. Afkhami-Jeddi et al.~\cite[Appendix~B]{MR4257711} already excluded the Leech lattice using the numerical spinning modular bootstrap. We give a conceptual, lattice-theoretic proof of this known exclusion: the classification of trace-zero involutions in $\mathrm{Co}_0$ shows that no such involution can arise from an even $K$. This identifies the obstruction directly in the symmetries and parity of the lattice, without a numerical bootstrap computation. Thus the Narain moduli space attains the optimal lattice packing density in dimension $8$~\cite{MR1957723} and the conjecturally optimal lattice packing densities in dimensions $12$ and $16$, but provably cannot attain the optimal lattice packing density in dimension $24$, since the Leech lattice is the unique densest lattice in dimension $24$~\cite{MR2600869}.

The paper is organized as follows. Section~\ref{sec:narain} develops the lattice setup and recalls the parametrization of Narain lattices by $\mathrm{SO}^+(c,c;\mathbb{R})$. Sections~\ref{sec:c2} and~\ref{sec:c3} prove Theorems~\ref{thm:c2} and~\ref{thm:c3} via the exceptional isomorphisms at $c=2$ and $c=3$, respectively. Section~\ref{sec:dense} develops the characterization of Narain structures on Euclidean lattices and carries out the comparison with the dense lattice candidates in dimensions $8$ through $24$.

\section{Narain lattice}\label{sec:narain}
We use the split-signature form $J_c$ and the Euclidean norm defined in the introduction. Put
\[
    S_c=\frac{1}{\sqrt 2}
    \begin{pmatrix}
        I_c&I_c\\
        I_c&-I_c
    \end{pmatrix},
\]
so that $S_c$ is symmetric and orthogonal; in particular $S_c^{-1}=S_c$. We write $O(n)$ for the Euclidean orthogonal group, and
\[
O(c,c;\mathbb R)
=
\{M\in \GL_{2c}(\mathbb R):M^T J_cM=J_c\}
\]
for the orthogonal group of $J_c$, and define $SO^+(c,c;\mathbb R)$ to be its identity component.

Define the \emph{reference lattice}
\[
\Lambda_0:=S_c\mathbb Z^{2c}\subset\mathbb R^{c,c}.
\]
Writing $n=(a,b)\in\mathbb Z^c\oplus\mathbb Z^c$, we have $S_cn=\tfrac1{\sqrt2}(a+b,\,a-b)$, and for $n'=(a',b')$ a direct calculation gives
\[
(S_cn)^TJ_c(S_cn')
=\frac12\bigl((a+b)\cdot(a'+b')-(a-b)\cdot(a'-b')\bigr)
=a\cdot b'+a'\cdot b.
\]
In particular $(S_cn)^TJ_c(S_cn)=2\,a\cdot b\in2\mathbb Z$, so $\Lambda_0$ is an even integral lattice with respect to $J_c$. Since $|\det S_c|=1$, it is unimodular, hence self-dual. Thus $\Lambda_0$ is itself a Narain lattice.

By the classification of indefinite even unimodular lattices~\cite[Chapter~II, Section~5]{MR506372}, any two even self-dual lattices of signature $(c,c)$ are isometric. Consequently, every Narain lattice $\Lambda\subset\mathbb R^{c,c}$ can be written as
\[
\Lambda=M\Lambda_0,
\qquad
M\in O(c,c;\mathbb R).
\]
We may always take $M\in SO^+(c,c;\mathbb R)$. Indeed, $\Lambda_0$ splits as an orthogonal sum over the coordinate planes $\mathbb R e_i\oplus\mathbb R e_{c+i}$, and in each such plane the sign changes $x_i\mapsto -x_i$ and $y_i\mapsto -y_i$ preserve $\Lambda_0$: they correspond to $(a_i,b_i)\mapsto(-b_i,-a_i)$ and $(a_i,b_i)\mapsto(b_i,a_i)$, respectively. These maps lie in $O(c,c;\mathbb R)$ and, together with their product and the identity, represent all four connected components. Replacing $M$ by $MA$ for a suitable automorphism $A$ of $\Lambda_0$ leaves $\Lambda=M\Lambda_0$ unchanged and puts $M$ in $SO^+(c,c;\mathbb R)$.

With this normalization, the spectral gap of the Narain CFT associated with $\Lambda=M\Lambda_0$ is
\[
\Delta_1(\Lambda)
=
\frac12
\min_{x\in\Lambda_0\setminus\{0\}}\|Mx\|^2.
\]

\section{Narain CFT for \texorpdfstring{$c=2$}{c=2}}\label{sec:c2}
For $c=2$, we begin with an explicit realization of the exceptional Lie group isomorphism
\[
SO^+(2,2;\mathbb R)
\cong
\bigl(\SL_2(\mathbb R)\times \SL_2(\mathbb R)\bigr)/\{\pm(I,I)\}.
\]

\begin{lemma}\label{iso-c2}
Define a linear map $\Phi_2:\mathbb R^4\longrightarrow M_2(\mathbb R)$ by
\[
\Phi_2(x_1,x_2,x_3,x_4)
=
\frac1{\sqrt2}
\begin{pmatrix}
 x_1+x_3 & x_2+x_4\\
-x_2+x_4 & x_1-x_3
\end{pmatrix}.
\]
Then, for every $M\in SO^+(2,2;\mathbb R)$, there exist $g,h\in\SL_2(\mathbb R)$, unique up to simultaneous sign change, such that
\[
\Phi_2(Mx)=g\,\Phi_2(x)\,h^{-1}
\qquad\text{for all }x\in\mathbb R^4.
\]
\end{lemma}

\begin{proof}
First, the identity
\[
\|\Phi_2(x)\|_F^2=x_1^2+x_2^2+x_3^2+x_4^2
\]
is immediate from the definition, so $\Phi_2$ is a Euclidean isometry. Also,
\begin{align*}
    \det\Phi_2(x) &=\frac12\bigl((x_1+x_3)(x_1-x_3)-(x_2+x_4)(-x_2+x_4)\bigr) \\
    &=\frac12(x_1^2+x_2^2-x_3^2-x_4^2)=\frac12x^TJ_2x.
\end{align*}
Thus $\Phi_2$ identifies the quadratic space $(\mathbb R^4,J_2)$ with $M_2(\mathbb R)$ equipped with the quadratic form $2\det$.

Now let $\SL_2(\mathbb R)\times\SL_2(\mathbb R)$ act on $M_2(\mathbb R)$ by
\[
(g,h)\cdot X=gXh^{-1}.
\]
Since $\det g=\det h=1$, this action preserves $\det X$. Conjugating by $\Phi_2$ gives a homomorphism
\[
\rho:\SL_2(\mathbb R)\times\SL_2(\mathbb R)
\longrightarrow SO^+(2,2;\mathbb R),
\qquad
\rho(g,h)=\Phi_2^{-1}\circ(X\mapsto gXh^{-1})\circ\Phi_2.
\]
The image lies in the identity component because the source is connected.

If $\rho(g,h)=\mathrm{id}$, then $gXh^{-1}=X$ for all $X\in M_2(\mathbb R)$. Taking $X=I$ gives $g=h$, and then $gXg^{-1}=X$ for all $X$, so $g$ is scalar. Since $g\in\SL_2(\mathbb R)$, we get $g=h=\pm I$. Hence the kernel is $\{\pm(I,I)\}$.

Both $\SL_2(\mathbb R)\times\SL_2(\mathbb R)$ and $SO^+(2,2;\mathbb R)$ have dimension $6$. Since the kernel is discrete, the differential of $\rho$ is injective, hence an isomorphism. Therefore the image is an open subgroup of the connected group $SO^+(2,2;\mathbb R)$, and so it is the whole identity component. This proves the asserted representation and the uniqueness up to the kernel.
\end{proof}

We now prove $\Delta_1 \le 2/3$, stated in the introduction.

\begin{proof}[Proof of the bound in Theorem~\ref{thm:c2}]
Let $\Lambda=M\Lambda_0$ be a Narain lattice with $M\in SO^+(2,2;\mathbb R)$, as in Section~\ref{sec:narain}, so that
\[
2\Delta_1(\Lambda)
=
\min_{x\in \Lambda_0\setminus\{0\}}\|Mx\|^2.
\]
By Lemma~\ref{iso-c2}, choose $g,h\in\SL_2(\mathbb R)$ such that
\[
\Phi_2(Mx)=g\Phi_2(x)h^{-1}
\qquad (x\in\mathbb R^4).
\]
For $n=(n_1,n_2,n_3,n_4)^T\in\mathbb Z^4$, a direct calculation gives
\[
\Phi_2(S_2n)=
\begin{pmatrix}
 n_1&n_2\\
-n_4&n_3
\end{pmatrix},
\]
hence $\Phi_2(\Lambda_0)=M_2(\mathbb Z)$. We restrict the minimum to those $x\in \Lambda_0$ for which $\Phi_2(x)$ has rank one. The nonzero rank-one elements of $M_2(\mathbb Z)$ are exactly the matrices $uv^T$, with $u,v\in\mathbb Z^2\setminus\{0\}$: if $X\in M_2(\mathbb Z)$ has rank one, let $u\in\mathbb Z^2$ be a primitive vector spanning its column space; then each column of $X$ is an integral multiple of $u$, so $X=uv^T$ with $v\in\mathbb Z^2$. Therefore
\begin{align*}
\min_{x\in \Lambda_0\setminus\{0\}}\|Mx\|^2 &= \min_{x\in \Lambda_0\setminus\{0\}}\norm{\Phi_2(Mx)}_F^2
\le
\min_{u,v\in\mathbb Z^2\setminus\{0\}}
\|g(uv^T)h^{-1}\|_F^2\\
&=
\min_{u,v\in\mathbb Z^2\setminus\{0\}}
\|(gu)(h^{-T}v)^T\|_F^2\\
&=
\left(\min_{u\in\mathbb Z^2\setminus\{0\}}\|gu\|^2\right)
\left(\min_{v\in\mathbb Z^2\setminus\{0\}}\|h^{-T}v\|^2\right).
\end{align*}
Since $g,h^{-T}\in\SL_2(\mathbb R)$, both define two-dimensional Euclidean lattices of covolume $1$. The Hermite bound in dimension two~\cite{MR1957723} therefore gives
\[
\min_{u\neq0}\|gu\|^2\le \gamma_2,
\qquad
\min_{v\neq0}\|h^{-T}v\|^2\le \gamma_2,
\qquad
\gamma_2=\frac{2}{\sqrt3}.
\]
Thus
\[
\min_{x\in \Lambda_0\setminus\{0\}}\|Mx\|^2
\le \gamma_2^2=\frac43,
\]
and hence
\[
\Delta_1\le \frac12\cdot\frac43=\frac23. \qedhere
\]
\end{proof}

\begin{lemma}\label{narain2-rank1}
Let $\Lambda=M\Lambda_0$ be a Narain lattice for $c=2$ with $M\in SO^+(2,2;\mathbb R)$. If $x\in \Lambda_0\setminus\{0\}$ attains the minimum of $\|Mx\|^2$, then $\Phi_2(x)$ has rank one.
\end{lemma}

\begin{proof}
By the bound just proved, any minimum vector $x$ satisfies
\[
\|Mx\|^2\le \frac43.
\]
Since $M^TJ_2M=J_2$, we have
\[
|x^TJ_2x|=|(Mx)^TJ_2(Mx)|\le \|Mx\|^2\le \frac43.
\]
On the other hand, $x\in\Lambda_0$ and $\Lambda_0$ is even, so $x^TJ_2x\in 2\mathbb Z$. Hence $x^TJ_2x=0$. Since $2\det\Phi_2(x)=x^TJ_2x$, we get $\det\Phi_2(x)=0$. Finally $x\neq0$ and $\Phi_2$ is an isomorphism, so $\Phi_2(x)\neq0$. Thus $\Phi_2(x)$ is a nonzero singular $2\times2$ matrix, hence has rank one.
\end{proof}

Lemma~\ref{narain2-rank1} shows that the minimum is attained on rank-one matrices, so the inequality in the proof of Theorem~\ref{thm:c2} obtained by restricting the minimum to rank-one matrices is an equality.
\begin{corollary}\label{c2-formula}
Let $\Lambda=M\Lambda_0$ be a Narain lattice for $c=2$ with $M\in SO^+(2,2;\mathbb R)$, and let $g,h\in\SL_2(\mathbb R)$ be as in Lemma~\ref{iso-c2}. Then
\[
\Delta_1(\Lambda)
=\frac12\left(\min_{u\in\mathbb Z^2\setminus\{0\}}\|gu\|^2\right)
\left(\min_{v\in\mathbb Z^2\setminus\{0\}}\|h^{-T}v\|^2\right).
\]
\end{corollary}

Using this formula, we now determine the equality case in Theorem~\ref{thm:c2}.

\begin{proof}[Proof of the equality case in Theorem~\ref{thm:c2}]
By Corollary~\ref{c2-formula},
\[
2\Delta_1(\Lambda)
=\left(\min_{u\in\mathbb Z^2\setminus\{0\}}\|gu\|^2\right)
\left(\min_{v\in\mathbb Z^2\setminus\{0\}}\|h^{-T}v\|^2\right),
\]
and each factor is at most $\gamma_2=2/\sqrt3$. Hence $\Delta_1(\Lambda)=2/3$ if and only if both factors equal $\gamma_2$, that is, if and only if both $g\mathbb Z^2$ and $h^{-T}\mathbb Z^2$ are hexagonal, since the hexagonal lattice is the only planar lattice attaining $\gamma_2$~\cite{MR1957723}.

All such Narain lattices are equivalent under $O(2)\times O(2)$. Since any two hexagonal lattices with covolume $1$ differ by a rotation, for another equality case we have $g'=Rg\gamma$ and $h'^{-T}=R'h^{-T}\gamma'$ with $R,R'\in SO(2)$ and $\gamma,\gamma'\in\SL_2(\mathbb Z)$, and hence
\[
g'M_2(\mathbb Z)h'^{-1}=R\bigl(gM_2(\mathbb Z)h^{-1}\bigr)R'^{-1}.
\]
The map $X\mapsto RXR'^{-1}$ preserves the Frobenius norm and $\det$, so under $\Phi_2$ it is an element of $O(2)\times O(2)$.

For a representative, let $A_2=B\mathbb Z^2$ with $B=\begin{pmatrix}\sqrt2&-1/\sqrt2\\0&\sqrt{3/2}\end{pmatrix}$, and take $g=h=3^{-1/4}B$. Then $g\mathbb Z^2=3^{-1/4}A_2$ and $h^{-T}\mathbb Z^2=3^{1/4}A_2^*$ are both hexagonal, and $\Phi_2(\rho(g,h)\Lambda_0)=B\,M_2(\mathbb Z)\,B^{-1}$ has the basis $b_i(b_j^*)^T$ $(1\le i,j\le2)$, where $b_1,b_2$ and $b_1^*,b_2^*$ are the columns of $B$ and $B^{-T}$, that is, bases of $A_2$ and $A_2^*$. Writing $\Phi_2^{-1}\bigl(b_i(b_j^*)^T\bigr)=(x,y)$ and expressing $x,y$ in the basis $b_1^*,b_2^*$ and $x-y$ in the basis $b_1,b_2$, a direct computation gives
\[
\renewcommand{\arraystretch}{1.3}
\setlength{\arraycolsep}{9pt}
\begin{array}{c|ccc}
(i,j) & x & y & x-y\\\hline
(1,1) & (1,0) & (1,0) & (0,0)\\
(1,2) & (0,1) & (0,1) & (0,0)\\
(2,1) & (0,-1) & (-1,1) & (0,-1)\\
(2,2) & (1,-1) & (-1,0) & (1,0)
\end{array}
\]
In every row, all coordinates are integers, so $x,y\in A_2^*$ and $x-y\in A_2$.
Hence $\rho(g,h)\Lambda_0$ is contained in $\Lambda(A_2)=\{(x,y)\in A_2^*\oplus A_2^*:\ x-y\in A_2\}$, and since both lattices have covolume one, they coincide. Therefore, up to $O(2)\times O(2)$, equality holds precisely for $\Lambda(A_2)$, the Narain lattice of the $\SU(3)_1$ WZW model.
\end{proof}

Using further classical facts about two-dimensional lattices, we can also describe the local maxima of $\Delta_1$ on the Narain moduli space.

\begin{theorem}\label{narain2-local}
For $c=2$, $\Delta_1$ attains a local maximum at $M\Lambda_0$ precisely when $g\mathbb Z^2$ and $h^{-T}\mathbb Z^2$ are both hexagonal. In particular, every local maximum of $\Delta_1$ is global.
\end{theorem}

\begin{proof}
By Corollary~\ref{c2-formula},
\[
2\Delta_1(\Lambda)=\left(\min_{u\in\mathbb Z^2\setminus\{0\}}\|gu\|^2\right) \left(\min_{v\in\mathbb Z^2\setminus\{0\}}\|h^{-T}v\|^2\right).
\]
Since $\rho$ is a surjective local diffeomorphism (Lemma~\ref{iso-c2}), local maxima of $\Delta_1$ in $M$ correspond to local maxima in $(g,h)\in\SL_2(\mathbb R)\times\SL_2(\mathbb R)$, so the two factors vary independently over all unimodular planar lattices. As both factors are positive, the product has a local maximum precisely when each factor does. By Voronoi's theorem, the local maxima of the Hermite invariant are exactly the perfect and eutactic lattices, and in dimension two the only such lattice is the hexagonal lattice~\cite{MR1957723}. Therefore $\Delta_1$ attains a local maximum precisely when both lattices are hexagonal, and by the equality case of Theorem~\ref{thm:c2}, this happens exactly when $\Delta_1=2/3$.
\end{proof}

\section{Narain CFT for \texorpdfstring{$c=3$}{c=3}}\label{sec:c3}

For $c=3$, we prove the analogous theorem based on the exceptional Lie group isomorphism
\[
SO^+(3,3;\mathbb R)\cong\SL_4(\mathbb R)/\{\pm I\}.
\]

\begin{lemma}\label{iso-c3}
Let $e_1,\dots,e_4$ be the standard basis of $\mathbb R^4$, write $e_{ij}=e_i\wedge e_j$, and let
\[
(f_1,\dots,f_6)=(e_{23},e_{31},e_{12},e_{14},e_{24},e_{34}).
\]
Define a linear map $\Phi_3:\mathbb R^6\longrightarrow\bigwedge\nolimits^2\mathbb R^4$ by
\[
\Phi_3(x)=\frac1{\sqrt2}\sum_{i=1}^3\Bigl((x_i+x_{i+3})\,f_i+(x_i-x_{i+3})\,f_{i+3}\Bigr).
\]
Then, for every $M\in SO^+(3,3;\mathbb R)$, there exists $g\in\SL_4(\mathbb R)$, unique up to sign, such that
\[
\Phi_3(Mx)=g\cdot\Phi_3(x)\qquad\text{for all }x\in\mathbb R^6,
\]
where $g$ acts on $\bigwedge^2\mathbb R^4$ by $g\cdot(u\wedge v)=(gu)\wedge(gv)$.
\end{lemma}

\begin{proof}
Equip $\bigwedge^2\mathbb R^4$ with the Euclidean inner product for which $(f_1,\dots,f_6)$ is orthonormal. The identity
\[
\norm{\Phi_3(x)}^2=x_1^2+\dots+x_6^2
\]
is immediate from the definition, so $\Phi_3$ is a Euclidean isometry. Also, writing $\Omega=e_1\wedge e_2\wedge e_3\wedge e_4$, a direct calculation gives
\[
f_i\wedge f_j=0,
\qquad
f_{i+3}\wedge f_{j+3}=0,
\qquad
f_i\wedge f_{j+3}=\delta_{ij}\,\Omega
\qquad(1\le i,j\le3).
\]
Since bivectors commute under the wedge product, this yields
\[
\Phi_3(x)\wedge\Phi_3(y)=(x^TJ_3y)\,\Omega
\qquad(x,y\in\mathbb R^6).
\]
Thus $\Phi_3$ identifies the quadratic space $(\mathbb R^6,J_3)$ with $\bigwedge^2\mathbb R^4$ equipped with the symmetric bilinear form $B$ defined by $\omega\wedge\omega'=B(\omega,\omega')\,\Omega$.
 
Now let $\SL_4(\mathbb R)$ act on $\bigwedge^2\mathbb R^4$ by $g\cdot(u\wedge v)=(gu)\wedge(gv)$. Since $\det g=1$, for $\omega, \omega' \in \bigwedge^2\mathbb R^4$, we have
\[(g\cdot\omega)\wedge(g\cdot\omega')=(\det g)\,\omega\wedge\omega'
=\omega\wedge\omega',\] so this action preserves the quadratic form. Conjugating by $\Phi_3$ gives a homomorphism
\[
\rho:\SL_4(\mathbb R)\longrightarrow SO^+(3,3;\mathbb R),
\qquad
\rho(g)=\Phi_3^{-1}\circ(\omega\mapsto g\cdot\omega)\circ\Phi_3 .
\]
The image lies in the identity component because the source is connected.
 
If $\rho(g)=\mathrm{id}$, then $(gu)\wedge(gv)=u\wedge v$ for all $u,v\in\mathbb R^4$, so $g$ preserves every two-dimensional subspace of $\mathbb R^4$, hence every line, and $g$ is scalar. Since $g\in\SL_4(\mathbb R)$, we get $g=\pm I$. Hence the kernel is $\{\pm I\}$.
 
Both $\SL_4(\mathbb R)$ and $SO^+(3,3;\mathbb R)$ have dimension $15$. Since the kernel is discrete, the differential of $\rho$ is injective, hence an isomorphism. Therefore the image is an open subgroup of the connected group $SO^+(3,3;\mathbb R)$, and so it is the whole identity component. This proves the asserted representation and the uniqueness up to the kernel.
\end{proof}

We now prove $\Delta_1 \le 3/4$, stated in the introduction.

\begin{proof}[Proof of the bound in Theorem~\ref{thm:c3}]
Let $\Lambda=M\Lambda_0$ be a Narain lattice with $M\in SO^+(3,3;\mathbb R)$, as in Section~\ref{sec:narain}, so that
\[
2\Delta_1(\Lambda)=\min_{x\in\Lambda_0\smallsetminus\{0\}}\|Mx\|^2 .
\]
By Lemma~\ref{iso-c3}, choose $g\in\SL_4(\mathbb R)$ such that
\[
\Phi_3(Mx)=g\cdot\Phi_3(x)\qquad(x\in\mathbb R^6).
\]
For $n=(n_1,\dots,n_6)^T\in\mathbb Z^6$, we have
\begin{align*}
\Phi_3(S_3n)&=n_1\,e_2\wedge e_3+n_2\,e_3\wedge e_1+n_3\,e_1\wedge e_2 \\
&+n_4\,e_1\wedge e_4+n_5\,e_2\wedge e_4+n_6\,e_3\wedge e_4,
\end{align*}
hence $\Phi_3(\Lambda_0)=\bigwedge^2\mathbb Z^4$. We restrict the minimum to those $x\in\Lambda_0$ for which $\Phi_3(x)$ is decomposable. The nonzero decomposable elements of $\bigwedge^2\mathbb Z^4$ are exactly the $u\wedge v$ with $u,v\in\mathbb Z^4$ linearly independent: if $\omega\in\bigwedge^2\mathbb Z^4$ is nonzero and decomposable over $\mathbb R$, then the plane $P=\{w\in\mathbb R^4:w\wedge\omega=0\}$ is rational, and extending a basis $u_0,v_0$ of $P\cap\mathbb Z^4$ to a basis of $\mathbb Z^4$ shows that $\omega=m\,u_0\wedge v_0=(mu_0)\wedge v_0$ for some $m\in\mathbb Z$. Therefore
\begin{align*}
\min_{x\in\Lambda_0\smallsetminus\{0\}}\|Mx\|^2 &= \min_{x\in\Lambda_0\smallsetminus\{0\}}\norm{\Phi_3(Mx)}^2
\le\min_{\substack{u,v\in\mathbb Z^4\\u\wedge v\neq0}}\|g\cdot(u\wedge v)\|^2\\
&=\min_{\substack{u,v\in\mathbb Z^4\\u\wedge v\neq0}}\|(gu)\wedge(gv)\|^2 .
\end{align*}
Since $\|a\wedge b\|^2=\|a\|^2\|b\|^2-(a\cdot b)^2$, the right-hand side is the minimal squared covolume of a rank-two sublattice of $g\mathbb Z^4$, which is a four-dimensional Euclidean lattice of covolume $1$. The Hermite--Rankin bound in dimension four~\cite[Theorem~5]{MR55380}  therefore gives
\[
\min_{\substack{u,v\in\mathbb Z^4\\u\wedge v\neq0}}\|(gu)\wedge(gv)\|^2
\le\gamma_{4,2},
\qquad
\gamma_{4,2}=\frac32.
\]
Thus
\[
\min_{x\in\Lambda_0\smallsetminus\{0\}}\|Mx\|^2\le\frac32,
\]
and hence
\[
\Delta_1\le\frac12\cdot\frac32=\frac34. \qedhere
\]
\end{proof}

\begin{lemma}\label{narain3-decomp}
Let $\Lambda=M\Lambda_0$ be a Narain lattice for $c=3$ with $M\in SO^+(3,3;\mathbb R)$. If $x\in\Lambda_0\smallsetminus\{0\}$ attains the minimum of $\|Mx\|^2$, then $\Phi_3(x)$ is decomposable.
\end{lemma}

\begin{proof}
By the bound just proved, any minimum vector $x$ satisfies
\[
\|Mx\|^2\le\frac32 .
\]
Since $M^TJ_3M=J_3$, we have
\[
|x^TJ_3x|=|(Mx)^TJ_3(Mx)|\le\|Mx\|^2\le\frac32 .
\]
On the other hand, $x\in\Lambda_0$ and $\Lambda_0$ is even, so $x^TJ_3x\in2\mathbb Z$. Hence $x^TJ_3x=0$. Since $\Phi_3(x)\wedge\Phi_3(x)=(x^TJ_3x)\,\Omega$, we get $\Phi_3(x)\wedge\Phi_3(x)=0$. Finally $x\neq0$ and $\Phi_3$ is an isomorphism, so $\Phi_3(x)\neq0$. Thus $\Phi_3(x)$ is a nonzero bivector in $\bigwedge^2\mathbb R^4$ with $\Phi_3(x)\wedge\Phi_3(x)=0$, hence decomposable by the Pl\"ucker relation.
\end{proof}

Lemma~\ref{narain3-decomp} shows that the minimum is attained on decomposable bivectors, so the inequality in the proof of Theorem~\ref{thm:c3} obtained by restricting the minimum to decomposable bivectors is an equality.

\begin{corollary}\label{c3-formula}
Let $\Lambda=M\Lambda_0$ be a Narain lattice for $c=3$ with $M\in SO^+(3,3;\mathbb R)$, and let $g\in\SL_4(\mathbb R)$ be as in Lemma~\ref{iso-c3}. Then
\[
\Delta_1(\Lambda)
=\frac12\min_{\substack{u,v\in\mathbb Z^4\\u\wedge v\neq0}}\|(gu)\wedge(gv)\|^2,
\]
that is, half the minimal squared covolume of a rank-two sublattice of $g\mathbb Z^4$.
\end{corollary}

Using this formula, we now determine the equality case in Theorem~\ref{thm:c3}.

\begin{proof}[Proof of the equality case in Theorem~\ref{thm:c3}]
By Corollary~\ref{c3-formula}, $\Delta_1(\Lambda)=3/4$ if and only if $g\mathbb Z^4$ attains $\gamma_{4,2}$, that is, if and only if $g\mathbb Z^4$ is similar to $D_4$, since $D_4$ is the only lattice attaining $\gamma_{4,2}$~\cite[Theorem~5]{MR55380}.\footnote{Rankin shows that equality holds exactly for the forms equivalent to a multiple of $f^*(x)=x_1^2+x_2^2+x_3^2+x_4^2+x_1x_2+x_1x_3+x_1x_4+x_2x_3+x_3x_4$, and the Gram matrix of the basis $e_1+e_2,\ e_1+e_3,\ e_1+e_4,\ e_2+e_4$ of $D_4$ represents the quadratic form $2f^*$.}

All such Narain lattices are equivalent under $O(3)\times O(3)$. Any two covolume-one lattices similar to $D_4$ differ by an orthogonal map, which can be taken in $SO(4)$ because $D_4$ has a reflection automorphism. Hence for another equality case we have $g'=Rg\gamma$ with $R\in SO(4)$ and $\gamma\in\SL_4(\mathbb Z)$, and
\[
g'\cdot\textstyle\bigwedge^2\mathbb Z^4=R\cdot\bigl(g\cdot\bigwedge^2\mathbb Z^4\bigr).
\]
The map $\omega\mapsto R\cdot\omega$ preserves the Euclidean inner product and the wedge pairing, so under $\Phi_3$ it is an element of $O(3)\times O(3)$.

For a representative, let $D_4=\{x\in\mathbb Z^4:\sum_ix_i\in2\mathbb Z\}$ and take
\[
g=2^{-1/4}\begin{pmatrix}
 1& 0& 0&0\\
-1& 1& 0&0\\
 0&-1& 1&1\\
 0& 0&-1&1
\end{pmatrix}.
\]
Then $g\mathbb Z^4=2^{-1/4}D_4$, and $\Phi_3(\rho(g)\Lambda_0)=g\cdot\bigwedge^2\mathbb Z^4$ has the basis $2^{-1/2}\,d_i\wedge d_j$ $(1\le i<j\le4)$, where $d_1,\dots,d_4\in D_4$ are the columns of the integral matrix $2^{1/4}g$. Take $A_3$ as $D_3=\{x\in\mathbb Z^3:\sum_ix_i\in2\mathbb Z\}$, so that $A_3^*=\mathbb Z^3\cup(\mathbb Z+\tfrac12)^3$. Writing $\Phi_3^{-1}\bigl(2^{-1/2}\,d_i\wedge d_j\bigr)=(x,y)$, a direct computation gives
\[
\renewcommand{\arraystretch}{1.3}
\setlength{\arraycolsep}{9pt}
\begin{array}{c|ccc}
(i,j) & x & y & x-y\\\hline
(1,2) & \tfrac12(1,1,1) & \tfrac12(1,1,1) & (0,0,0)\\
(1,3) & (-1,0,0) & (0,-1,0) & (-1,1,0)\\
(1,4) & (0,-1,0) & (-1,0,0) & (1,-1,0)\\
(2,3) & \tfrac12(1,-1,1) & \tfrac12(1,1,-1) & (0,-1,1)\\
(2,4) & \tfrac12(1,1,-1) & \tfrac12(1,-1,1) & (0,1,-1)\\
(3,4) & (0,0,1) & (0,0,-1) & (0,0,2)
\end{array}
\]
In every row, $x,y\in A_3^*$ and $x-y\in A_3$.
Hence $\rho(g)\Lambda_0$ is contained in $\Lambda(A_3)=\{(x,y)\in A_3^*\oplus A_3^*:\ x-y\in A_3\}$, and since both lattices have covolume one, they coincide. Therefore, up to $O(3)\times O(3)$, equality holds precisely for $\Lambda(A_3)$, the Narain lattice of the $\mathrm{SU}(4)_1$ WZW model.
\end{proof}

Using further results on rank-two sublattices, we can also describe the local maxima of $\Delta_1$ on the Narain moduli space. We first record an elementary lemma, where the determinant of a lattice means the determinant of its Gram matrix.

\begin{lemma}\label{even-rank2}
Let $L\subset\mathbb R^4$ be a full-rank even integral lattice. Then every rank-two sublattice $N\subset L$ satisfies $\det N\ge3$, with equality if and only if $N$ is isometric to $A_2$. Consequently, if $L$ contains a sublattice isometric to $A_2$ and $d=\det L$, then the minimal squared covolume of a rank-two sublattice of the covolume-one lattice $d^{-1/8}L$ is $3/\sqrt d$.
\end{lemma}

\begin{proof}
Choose a reduced basis $(u,v)$ of $N$, so that $\|u\|^2\le\|v\|^2$ and $2|u\cdot v|\le\|u\|^2$~\cite{MR1957723}. Put $a=\|u\|^2$, $b=u\cdot v$, $c=\|v\|^2$. Since $L$ is even integral, $a,c\in2\mathbb Z_{>0}$ and $b\in\mathbb Z$, hence $\det N=ac-b^2\ge3a^2/4\ge3$. Equality forces $a=c=2$ and $|b|=1$, the Gram matrix of $A_2$. Finally, rescaling by $d^{-1/8}$ makes the covolume one and multiplies the determinant of every rank-two sublattice by $d^{-1/2}$.
\end{proof}

For a lattice $L\subset\mathbb R^n$, consider the minimum of $\det N/(\det L)^{2/n}$ over rank-two sublattices $N\subset L$; the planes spanned by sublattices attaining it are the \emph{minimal planes} of $L$. Following~\cite{MR1407461}, $L$ is \emph{$2$-extreme} if this quantity has a local maximum at $L$, \emph{$2$-perfect} if the orthogonal projections onto its minimal planes span the space of symmetric endomorphisms of $\mathbb R^n$, and \emph{$2$-eutactic} if the identity is a positive linear combination of these projections.

\begin{theorem}\label{c3-local}
For $c=3$, $\Delta_1$ attains a local maximum at $M\Lambda_0$ precisely when $g\mathbb Z^4$ is $2$-extreme. Moreover, $A_4$ rescaled to covolume one is $2$-extreme and gives a local maximum with $\Delta_1=3/(2\sqrt5)<3/4$. In particular, in contrast with the case $c=2$, not every local maximum of $\Delta_1$ is global.
\end{theorem}

\begin{proof}
By Corollary~\ref{c3-formula}, $2\Delta_1$ is the minimal squared covolume of a rank-two sublattice of $g\mathbb Z^4$. Since $\rho$ is a surjective local diffeomorphism (Lemma~\ref{iso-c3}), local maxima of $\Delta_1$ in $M$ correspond to local maxima in $g\in\SL_4(\mathbb R)$, so $g\mathbb Z^4$ varies over all unimodular lattices in $\mathbb R^4$. The local maxima of the minimal squared covolume of a rank-two sublattice are the $2$-extreme lattices, which are exactly the $2$-perfect and $2$-eutactic lattices~\cite[Theorem~3.2.3]{MR1407461}. Therefore $\Delta_1$ attains a local maximum precisely when $g\mathbb Z^4$ is $2$-extreme.

Coulangeon also showed that every irreducible root lattice of rank at least $4$ is $2$-extreme~\cite[Theorem~5.1.1]{MR1407461}; in particular, $A_4$ is $2$-extreme. Since $A_4$ is even of determinant $5$ and contains $A_2$, Lemma~\ref{even-rank2} shows that the minimal squared covolume of a rank-two sublattice of $A_4$ rescaled to covolume one is $3/\sqrt5$, so $\Delta_1=3/(2\sqrt5)$.
\end{proof}

Theorem~\ref{c3-local} is not a classification: to our knowledge, the $2$-extreme lattices in dimension four have not been classified (see~\cite[Section~5.2]{MR1407461} and~\cite[Appendix, \S A6, Table~1, p.~80]{watanabe-fundamental} for the $2$-perfect lattices currently known in dimension four; not all of them are $2$-eutactic, e.g., $W_4$ and $W_4^*$), so further local maxima may exist. The dual lattice $A_4^*$, however, gives nothing new. The Hodge star $\ast$ on $\bigwedge^2\mathbb R^4$ satisfies $\ast f_i=f_{i+3}$ and $\ast f_{i+3}=f_i$, so $\Phi_3^{-1}\circ\ast\circ\Phi_3=J_3\in O(3)\times O(3)$. Moreover, one checks that $\ast(g\cdot\omega)=g^{-T}\cdot(\ast\omega)$ for $g\in\SL_4(\mathbb R)$, so $\ast$ maps $\bigwedge^2(g\mathbb Z^4)$ onto $\bigwedge^2\bigl((g\mathbb Z^4)^*\bigr)$. Hence the Narain lattices associated with a covolume-one lattice and with its dual differ by $J_3$, and $A_4^*$ gives the same Narain CFT as $A_4$.

\section{Comparison with dense Euclidean lattices}\label{sec:dense}

A Narain lattice of central charge $c$, when equipped with the Euclidean inner product, is a determinant-one Euclidean lattice of dimension $2c$. Therefore,
\[
    2\Delta_1(\Lambda)
    =
    \min_{x\in\Lambda\setminus\{0\}}\lVert x\rVert^2
    \leq
    \gamma_{2c},
\]
where $\gamma_{2c}$ is the Hermite constant in dimension $2c$.

For $4\leq c\leq12$, let $L_{2c}$ denote the lattice listed in the corresponding row of Table~\ref{tab:dense-lattice-obstruction}, the densest lattice known in dimension $2c$~\cite[Table~1.2]{MR1662447}, normalized to have determinant one, and let
\[
    \Gamma_{2c}
    =
    \min_{x\in L_{2c}\setminus\{0\}}\lVert x\rVert^2
\]
denote its squared minimum. If $L_{2c}$ admits a Narain structure, then the corresponding Narain lattice has spectral gap $\Delta_1=\Gamma_{2c}/2$. We now determine which of these dense Euclidean candidates admit such a structure.

The following criterion is an integral form of~\cite[Lemma~B.1]{MR4257711}, which shows that symmetric matrices $G$ and $H$, with $G$ positive definite and $H$ of signature $(c,c)$, are the Gram matrices $B^TB$ and $B^TJ_cB$ of a common basis $B$ if and only if $(GH^{-1})^2=I$. We include a short proof for completeness.

\begin{theorem}\label{narain-characterization}
    Let $E\in\operatorname{Sym}_{2c}^{+}(\mathbb R)$ with $\det E=1$. Then $E$ is the Euclidean Gram matrix of a Narain lattice if and only if there exists a symmetric integral matrix $K$ with even diagonal entries such that
    \[
        KE^{-1}K=E, \qquad \tr(E^{-1}K)=0.
    \]
\end{theorem}

\begin{proof}
    Suppose first that $E$ is the Euclidean Gram matrix of a Narain lattice $\Lambda$, and let $B\in GL_{2c}(\mathbb R)$ be a basis matrix for $\Lambda$, so that $E=B^TB$. Set $K=B^TJ_cB$. Then $K$ is an even symmetric integral matrix since $\Lambda$ is a Narain lattice. Moreover,
    \[
        KE^{-1}K
        =
        B^TJ_cB(B^TB)^{-1}B^TJ_cB
        =
        B^TJ_c^2B
        =
        B^TB
        =
        E.
    \]
    We also have $E^{-1}K=B^{-1}J_cB$, and hence $\tr(E^{-1}K)=\tr J_c=0$.

    Conversely, suppose that such a matrix $K$ exists, and let $A=E^{-1/2}KE^{-1/2}$. Then $A$ is symmetric and $A^2=I$. Hence all eigenvalues of $A$ are equal to $\pm1$. Furthermore, $\tr A=\tr(E^{-1}K)=0$, so each eigenvalue occurs with multiplicity $c$. By the spectral theorem, there exists $Q\in\Orth(2c)$ such that $A=Q^TJ_cQ$.

    Let $B=QE^{1/2}$ and $\Lambda=B\mathbb Z^{2c}$. Then $B^TB=E$ and
    \[
        B^TJ_cB
        =
        E^{1/2}Q^TJ_cQE^{1/2}
        =
        E^{1/2}AE^{1/2}
        =
        K.
    \]
    Thus $\Lambda$ is even and integral. Finally, taking determinants in $KE^{-1}K=E$ gives $(\det K)^2=(\det E)^2=1$. Therefore $K$ is unimodular, so $\Lambda$ is self-dual. Hence $\Lambda$ is a Narain lattice with Euclidean Gram matrix $E$.
\end{proof}

\subsection{An arithmetic obstruction}

For $4\leq c\leq12$, choose an integral Gram matrix $\widetilde E_c$ of a lattice similar to $L_{2c}$ that is primitive, i.e., whose entries have greatest common divisor $1$, and let $E_c$ be its determinant-one normalization. Let $d_c=\det\widetilde E_c$, and let $q_c$ be the smallest positive integer such that
\[
    G_c=q_c\widetilde E_c^{-1}
\]
is integral. Then
\[
    E_c=d_c^{-1/(2c)}\widetilde E_c,
    \qquad
    E_c^{-1}=\frac{d_c^{1/(2c)}}{q_c}G_c.
\]

By Theorem~\ref{narain-characterization}, if $E_c$ is the Euclidean Gram matrix of a Narain lattice, then there exists an even symmetric integral matrix $K$ satisfying
\[
    KG_cK=q_cd_c^{-1/c}\widetilde E_c,
    \qquad
    \tr(G_cK)=0.
\]
Since the left-hand side of the first equation is integral and $\widetilde E_c$ is primitive, it follows that $q_cd_c^{-1/c}$ must be an integer.

For the standard dense lattice candidates, the following table lists these quantities, computed from the integral Gram matrices in the Nebe--Sloane catalogue of lattices~\cite{NebeSloane} (see also~\cite{MR1662447}).

\begin{table}[ht]
    \centering
    \begin{tabular}{c|c|c|c|c}
        $c$ & $L_{2c}$ & $d_c$ & $q_c$ & $q_cd_c^{-1/c}$ \\ \hline
        $4$  & $E_8$          & $1$   & $1$  & $1$ \\
        $5$  & $\Lambda_{10}$ & $768$ & $12$ & $\displaystyle \frac{12}{\sqrt[5]{768}}$ \\
        $6$  & $K_{12}$       & $3^6$ & $3$  & $1$ \\
        $7$  & $\Lambda_{14}$ & $768$ & $12$ & $\displaystyle \frac{12}{\sqrt[7]{768}}$ \\
        $8$  & $BW_{16}$      & $2^8$ & $2$  & $1$ \\
        $9$  & $\Lambda_{18}$ & $192$ & $6$  & $\displaystyle \frac{6}{\sqrt[9]{192}}$ \\
        $10$ & $\Lambda_{20}$ & $2^6$ & $4$  & $2^{7/5}$ \\
        $11$ & $\Lambda_{22}$ & $12$  & $6$  & $\displaystyle \frac{6}{\sqrt[11]{12}}$ \\
        $12$ & $\Lambda_{24}$ & $1$   & $1$  & $1$
    \end{tabular}
    \caption{The arithmetic obstruction for the standard dense Euclidean lattice candidates. Here $K_{12}$ is the Coxeter--Todd lattice, $BW_{16}$ the Barnes--Wall lattice, and $\Lambda_n$ the laminated lattices.}
    \label{tab:dense-lattice-obstruction}
\end{table}

For $c=5,7,9,10,11$, the quantity $q_cd_c^{-1/c}$ is not an integer (Table~\ref{tab:dense-lattice-obstruction}). Therefore the corresponding putatively densest Euclidean lattices cannot occur as the Euclidean lattices underlying Narain lattices.

The arithmetic obstruction leaves the cases $c=4,6,8,12$. The lattices $E_8$, $K_{12}$, and $BW_{16}$ are known to admit Narain structures, and we computationally show that each admits a unique such structure up to automorphisms of the underlying Euclidean lattice. The remaining case $c=12$, corresponding to the Leech lattice, is ruled out separately.

\subsection{The cases \texorpdfstring{$c=4,6,8$}{c=4, 6, 8}}

For a fixed determinant-one Euclidean Gram matrix $E$, define
\[
    \mathcal K(E)=\bigl\{K\in\Sym_{2c}(\mathbb Z):K\text{ is even},\ KE^{-1}K=E,\ \tr(E^{-1}K)=0\bigr\}.
\]
Fix a basis matrix $B\in GL_{2c}(\mathbb R)$ with $E=B^TB$, and let $\Lambda=B\mathbb Z^{2c}$ be the Euclidean lattice it spans. Define
\[
    \Aut(\Lambda)=\{U\in O(2c):U\Lambda=\Lambda\}
\]
and
\[
\Orth(E,\mathbb Z)=\{P\in GL_{2c}(\mathbb Z):P^TEP=E\}.
\]
The latter is the coordinate realization of $\Aut(\Lambda)$ in the basis $B$, as the map $P\mapsto BPB^{-1}$ gives an isomorphism $\Orth(E,\mathbb Z)\cong\Aut(\Lambda)$. Some authors call either of these groups the automorphism group of the lattice; we reserve $\Aut(\Lambda)$ for the geometric group and use $\Orth(E,\mathbb Z)$ for its integral realization, which is more convenient for the enumeration below.

The group $\Orth(E,\mathbb Z)$ acts on the right on $\mathcal K(E)$ by $K\cdot P=P^TKP$. Indeed, if $K\in\mathcal K(E)$ and $P\in\mathrm O(E,\mathbb Z)$, then $P^TKP$ is symmetric and integral, and it is even since $x^TP^TKPx=(Px)^TK(Px)\in2\mathbb Z$ for every $x\in\mathbb Z^{2c}$. Moreover, $P^TEP=E$ implies $PE^{-1}P^T=E^{-1}$, and hence
\[
    (P^TKP)E^{-1}(P^TKP)=P^TKE^{-1}KP=P^TEP=E.
\]
The trace condition is also preserved, since
\[
    \tr(E^{-1}P^TKP)=\tr(PE^{-1}P^TK)=\tr(E^{-1}K)=0.
\]
Thus $P^TKP\in\mathcal K(E)$. Finally, $(K\cdot P)\cdot Q=Q^TP^TKPQ=K\cdot(PQ)$, so this is a right action.

This action corresponds to conjugation of the associated involutions. If $K'=K\cdot P=P^TKP$, then $x^TK'y=(Px)^TK(Py)$ for all $x,y\in\mathbb Z^{2c}$. Since $P$ represents an isometry of the Euclidean lattice $\Lambda$, the Narain structure defined by $K'$ is obtained from that defined by $K$ by applying a symmetry of $\Lambda$. We therefore regard two elements of $\mathcal K(E)$ in the same $\Orth(E,\mathbb Z)$-orbit as defining equivalent Narain structures on $\Lambda$. This agrees with the $O(c)\times O(c)$ equivalence of Narain lattices: if $B_i^TB_i=E$ and $B_i^TJ_cB_i=K_i$ for $i=1,2$ with $K_2=P^TK_1P$, then $Q=B_1PB_2^{-1}$ satisfies $Q^TQ=I$ and $Q^TJ_cQ=J_c$, so $Q\in O(c)\times O(c)$ and $QB_2\mathbb Z^{2c}=B_1\mathbb Z^{2c}$; conversely, such a $Q$ yields such a $P=B_1^{-1}QB_2$.

\begin{theorem}\label{thm:unique-structures}
    Each of the Euclidean lattices $E_8$, $K_{12}$, and $BW_{16}$ admits, up to automorphisms of the underlying Euclidean lattice, a unique Narain structure.
\end{theorem}

\begin{proof}
    The existence of Narain structures on $E_8$, $K_{12}$, and $BW_{16}$ was established in~\cite[Table~1 and Appendix~B]{MR4257711}. It remains to determine the number of equivalence classes.

    For each lattice, choose matrices $\widetilde E_c$ and $G_c$ as above. We enumerate all even symmetric integral solutions $K$ to
    \[
        KG_cK=q_cd_c^{-1/c}\widetilde E_c,
        \qquad
        \tr(G_cK)=0.
    \]
    In each of the cases $c=4,6,8$, we have $q_cd_c^{-1/c}=1$, so these equations reduce to
    \[
        KG_cK=\widetilde E_c,
        \qquad
        \tr(G_cK)=0.
    \]

    We start from the Cartan matrix of $E_8$ and the Gram matrices of $K_{12}$ and $BW_{16}$ in~\cite{NebeSloane}. The short vectors are enumerated with PARI/GP~\cite{PARI} (\texttt{qfminim}), and this list is checked against an independent exact enumeration based on a rational $LDL^T$ decomposition with integer coordinate bounds; all subsequent steps use exact integer arithmetic. After enumerating all solutions, we quotient them by the congruence action $K\mapsto P^TKP$ of $\Orth(\widetilde E_c,\mathbb Z)$, whose generators are computed with \texttt{qfauto} and verified to satisfy $P^T\widetilde E_cP=\widetilde E_c$. The enumeration itself proceeds column by column. Writing $k_j$ for the $j$-th column of $K$, the equation $KG_cK=\widetilde E_c$ reads entrywise
\[
    k_i^{T}G_c\,k_j=(\widetilde E_c)_{ij}.
\]
In particular, taking $i=j$ shows that each column $k_j$ is an integral vector of $G_c$-norm
\[
    k_j^{T}G_c\,k_j=(\widetilde E_c)_{jj},
\]
whose $j$-th coordinate is even. Enumerating once, with the Fincke--Pohst algorithm (\texttt{qfminim}), the finite ball
\[
    \bigl\{\,v\in\mathbb Z^{2c}\setminus\{0\}
        : v^{T}G_c v\le\max_j(\widetilde E_c)_{jj}\,\bigr\}
\]
therefore supplies, for each $j$, the finite list
\[
    C_j=\bigl\{\,v\in\mathbb Z^{2c}: v^{T}G_c v=(\widetilde E_c)_{jj},\ v_j\in2\mathbb Z\,\bigr\}
\]
of candidate columns.

We then fix the columns one at a time. Having chosen $k_{\mathrm{col}}$, we retain a vector $u$ as a candidate for column $j$ only when
\[
    u^{T}G_c\,k_{\mathrm{col}}=(\widetilde E_c)_{j,\mathrm{col}},
    \qquad
    u_{\mathrm{col}}=(k_{\mathrm{col}})_j,
    \qquad
    u_j\equiv 0 \ (\mathrm{mod}\ 2),
\]
the three conditions enforcing the Gram equation, the symmetry of $K$, and the even diagonal, respectively. Since
\[
    \tr(G_cK)=\sum_j (G_c k_j)_j,
\]
a partial assignment is abandoned as soon as the minimal and maximal remaining contributions show that $\tr(G_cK)=0$ can no longer be met, and a completed $K$ is kept precisely when $\tr(G_cK)=0$.

This produces the complete set of even symmetric integral solutions. We recheck every solution against the defining identities in exact arithmetic, and then perform the congruence quotient described above.

    Finally, for one solution $K_0$ we compute its orbit under $\Orth(\widetilde E_c,\mathbb Z)$ by applying the generators until the orbit closes. The results are listed in Table~\ref{tab:enumeration}.

\begin{table}[ht]
    \centering
    \begin{tabular}{c|c|c|c|c}
         & $|\Orth(\widetilde E_c,\mathbb Z)|$ & $|C_j|$ & $\#\{K\}$ & $\#$orbits \\ \hline
        $E_8$    & $696\,729\,600$        & $128$     & $3\,150$  & $1$ \\
        $K_{12}$ & $78\,382\,080$         & $436$     & $756$     & $1$ \\
        $BW_{16}$ & $89\,181\,388\,800$   & $2\,400$  & $17\,280$ & $1$
    \end{tabular}
    \caption{Results of the enumeration in the proof of Theorem~\ref{thm:unique-structures}. The order of $\Orth(\widetilde E_c,\mathbb Z)$ is as reported by \texttt{qfauto}, and the number $|C_j|$ of candidate columns is the same for all $j$ in the chosen bases.}
    \label{tab:enumeration}
\end{table}

In each case the orbit of $K_0$ under the group generated by these verified automorphisms is already the full set of solutions, so $\Orth(\widetilde E_c,\mathbb Z)$ acts transitively on it. Hence each lattice admits exactly one Narain structure up to automorphisms of the underlying Euclidean lattice. The program, the input Gram matrices, and the solutions are provided as ancillary files with the arXiv version of this paper.
\end{proof}

Consequently, the standard dense Euclidean candidates attain the corresponding Euclidean packing values within the Narain moduli space for $c=4,6,8$; the resulting Narain lattices have $\Delta_1=1$, $2/\sqrt3$, and $\sqrt2$, respectively. In these cases,
\[
    \sup_{\Lambda\subset\R^{c,c}\ \mathrm{Narain}}\Delta_1(\Lambda)\ge \frac{\Gamma_{2c}}{2}.
\]

\subsection{The Leech lattice}

Afkhami-Jeddi et al.~\cite[Appendix~B]{MR4257711} showed that the Leech lattice cannot underlie a Narain CFT, as it is excluded by their numerical spinning-bootstrap bounds. We give a conceptual proof of the same result using the automorphism group of the Leech lattice and the evenness of the Narain form. The arithmetic obstruction above does not suffice, since $q_{12}d_{12}^{-1/12}=1$; the additional obstruction comes from the parity condition on $K$.

\begin{theorem}\label{thm:leech}
    The Leech lattice cannot occur as the Euclidean lattice underlying a Narain lattice.
\end{theorem}

\begin{proof}
    Choose a basis matrix $B$ of the Leech lattice $\Lambda_{24}=B\mathbb Z^{24}$, let $E=B^TB$, and suppose for contradiction that $E$ is the Euclidean Gram matrix of a Narain lattice. By Theorem~\ref{narain-characterization}, there exists an even symmetric integral matrix $K$ such that $KE^{-1}K=E$ and $\tr(E^{-1}K)=0$.

    Let $T=E^{-1}K$. Since $E$ and $K$ are symmetric, we have
    \[
        T^TET=KE^{-1}K=E.
    \]
    Since the Leech lattice is unimodular, $E^{-1}$ is integral. Hence $T$ is integral, and the identity $T^TET=E$ shows that $T\in\Orth(E, \mathbb Z)$. Therefore, if we define $S = BTB^{-1}$, then $S \in \Aut(\Lambda_{24})$. For every $v=Bx \in \Lambda_{24}$, we have
    \[
        v^T S v = (Bx)^T S (Bx)=x^T(B^T B)Tx=x^TKx\in2\mathbb Z
    \]
    since $K$ is even. Furthermore,
    \[
        S^2 = BT^2B^{-1} = BE^{-1}KE^{-1}KB^{-1} = BE^{-1}EB^{-1} = I
    \] and $\tr (S) = \tr(T) = 0$.

    We use the standard coordinates for $\Lambda_{24}$~\cite[Chapter~4, \S11]{MR1662447}: a vector $\frac1{\sqrt8}x$ with $x\in\mathbb Z^{24}$ lies in $\Lambda_{24}$ if and only if, for some $m\in\{0,1\}$, all $x_i\equiv m\pmod2$, $\sum_ix_i\equiv4m\pmod8$, and for each $k$ the set $\{i:x_i\equiv k\pmod4\}$ is a codeword of the extended binary Golay code. In these coordinates, for each codeword $C$, the sign change $\varepsilon_C$, which negates the coordinates indexed by $C$ and fixes the remaining ones, is an automorphism of $\Lambda_{24}$. By the classification of involutions in $\mathrm{Co}_0=\Aut(\Lambda_{24})$~\cite[Theorem~10.15]{MR1707296}, every involution in $\Aut(\Lambda_{24})$ is conjugate to $\varepsilon_C$ for a codeword $C$ that is an octad, the complement of an octad, a dodecad, or the whole coordinate set; octads and dodecads are the codewords of weight $8$ and $12$. Since $\tr\varepsilon_C=24-2|C|$, the trace-zero involution $S$ is conjugate to $\varepsilon_D$ for some dodecad $D$. Choose $U\in\Aut(\Lambda_{24})$ with $U S U^{-1}=\varepsilon_D$. For any $v\in\Lambda_{24}$ we have $U^{-1}v\in\Lambda_{24}$, and $U^T=U^{-1}$, so
    \[
        v^T\varepsilon_Dv=v^TUSU^{-1}v=\bigl(U^{-1}v\bigr)^TS\bigl(U^{-1}v\bigr)
        \in2\mathbb Z.
    \]

    Now fix $j\in D$ and let
    \[
        v=\tfrac{1}{\sqrt8}\,(1,\dots,1,-3,1,\dots,1)\in\Lambda_{24},
    \]
    where the entry $-3$ occurs in the $j$-th coordinate.\footnote{This is a minimal vector of $\Lambda_{24}$: all coordinates are odd and congruent to $1$ modulo $4$, so the sets $\{i:x_i\equiv k\pmod 4\}$ are the whole coordinate set and the empty set, both Golay codewords; the coordinate sum is $20\equiv4\pmod 8$; and the squared norm is $(23+9)/8=4$.} Among the twelve coordinates negated by $\varepsilon_D$ one is the $j$-th, carrying the entry $-3$, while the twelve fixed coordinates all carry the entry $1$. Hence
    \[
        v^T\varepsilon_Dv=\frac{12-(11+9)}{8}=-1,
    \]
    which is odd, giving a contradiction. Therefore no such $K$ exists, and the Leech lattice cannot occur as the Euclidean lattice underlying a Narain lattice.
\end{proof}

We have therefore determined which of the standard dense Euclidean lattice candidates in dimensions $8$ through $24$ admit Narain structures. The candidates in dimensions $10$, $14$, $18$, $20$, and $22$ are excluded by the arithmetic obstruction. By the computer-assisted Theorem~\ref{thm:unique-structures}, the lattices $E_8$, $K_{12}$, and $BW_{16}$ each admit a unique Narain structure up to automorphisms of the underlying Euclidean lattice, while the Leech lattice is excluded by the classification of its trace-zero involutions.

\section*{Acknowledgments}
We are grateful to our advisor, Henry Cohn, for suggesting this problem, for his guidance and many helpful discussions, for advice on the computations and the literature, and for comments on earlier drafts of this paper.

During the preparation of this paper, the authors used Claude Opus~5.5 (Anthropic) as an assistant for editing the exposition and the \LaTeX{} source, for suggestions on the organization of the paper, and for checking the rigor of the proofs and proposing arguments for some auxiliary steps. All mathematical content, computations, and references were independently verified by the authors, who take full responsibility for the contents of this paper.

\bibliographystyle{plain}
\bibliography{references}

\begin{thebibliography}{10}

\bibitem{MR4257711}
Nima Afkhami-Jeddi, Henry Cohn, Thomas Hartman, and Amirhossein Tajdini.
\newblock Free partition functions and an averaged holographic duality.
\newblock {\em J. High Energy Phys.}, 2021(1):Paper No. 130, 42, 2021.

\bibitem{MR4916198}
Nikolaos Angelinos.
\newblock Code construction and ensemble holography of simply-laced {WZW}
  models at level 1.
\newblock {\em J. High Energy Phys.}, 2025(6):Paper No. 2, 30, 2025.

\bibitem{MR4514857}
Nikolaos Angelinos, Debarghya Chakraborty, and Anatoly Dymarsky.
\newblock Optimal {N}arain {CFT}s from codes.
\newblock {\em J. High Energy Phys.}, 2022(11):Paper No. 118, 21, 2022.

\bibitem{MR4327142}
Nathan Benjamin, Scott Collier, A.~Liam Fitzpatrick, Alexander Maloney, and
  Eric Perlmutter.
\newblock Harmonic analysis of 2d {CFT} partition functions.
\newblock {\em J. High Energy Phys.}, 2021(9):Paper No. 174, 50, 2021.

\bibitem{MR1973059}
Henry Cohn and Noam Elkies.
\newblock New upper bounds on sphere packings. {I}.
\newblock {\em Ann. of Math. (2)}, 157(2):689--714, 2003.

\bibitem{MR2600869}
Henry Cohn and Abhinav Kumar.
\newblock Optimality and uniqueness of the {L}eech lattice among lattices.
\newblock {\em Ann. of Math. (2)}, 170(3):1003--1050, 2009.

\bibitem{MR1662447}
J.~H. Conway and N.~J.~A. Sloane.
\newblock {\em Sphere packings, lattices and groups}, volume 290 of {\em
  Grundlehren der mathematischen Wissenschaften [Fundamental Principles of
  Mathematical Sciences]}.
\newblock Springer-Verlag, New York, third edition, 1999.
\newblock With additional contributions by E. Bannai, R. E. Borcherds, J.
  Leech, S. P. Norton, A. M. Odlyzko, R. A. Parker, L. Queen and B. B. Venkov.

\bibitem{MR1407461}
Renaud Coulangeon.
\newblock R\'eseaux {$k$}-extr\^emes.
\newblock {\em Proc. London Math. Soc. (3)}, 73(3):555--574, 1996.

\bibitem{MR1467163}
Renaud Coulangeon.
\newblock Minimal vectors in the second exterior power of a lattice.
\newblock {\em J. Algebra}, 194(2):467--476, 1997.

\bibitem{MR4771416}
A.~Liam Fitzpatrick and Wei Li.
\newblock Improving modular bootstrap bounds with integrality.
\newblock {\em J. High Energy Phys.}, 2024(7):Paper No. 58, 16, 2024.

\bibitem{MR1288936}
A.~Giveon, M.~Porrati, and E.~Rabinovici.
\newblock Target space duality in string theory.
\newblock {\em Phys. Rep.}, 244(2-3):77--202, 1994.

\bibitem{MR1707296}
Robert~L. Griess, Jr.
\newblock {\em Twelve sporadic groups}.
\newblock Springer Monographs in Mathematics. Springer-Verlag, Berlin, 1998.

\bibitem{MR4075697}
Thomas Hartman, Dalimil Maz\'a\v{c}, and Leonardo Rastelli.
\newblock Sphere packing and quantum gravity.
\newblock {\em J. High Energy Phys.}, 2019(12):Paper No. 048, 66, 2019.

\bibitem{MR1957723}
Jacques Martinet.
\newblock {\em Perfect lattices in {E}uclidean spaces}, volume 327 of {\em
  Grundlehren der mathematischen Wissenschaften [Fundamental Principles of
  Mathematical Sciences]}.
\newblock Springer-Verlag, Berlin, 2003.

\bibitem{MR506372}
John Milnor and Dale Husemoller.
\newblock {\em Symmetric bilinear forms}, volume~73 of {\em Ergebnisse der
  Mathematik und ihrer Grenzgebiete [Results in Mathematics and Related
  Areas]}.
\newblock Springer-Verlag, New York-Heidelberg, 1973.

\bibitem{MR834338}
K.~S. Narain.
\newblock New heterotic string theories in uncompactified dimensions {$<10$}.
\newblock {\em Phys. Lett. B}, 169(1):41--46, 1986.

\bibitem{NebeSloane}
Gabriele Nebe and N.~J.~A. Sloane.
\newblock A catalogue of lattices.
\newblock \url{https://www.math.rwth-aachen.de/~Gabriele.Nebe/LATTICES/}.

\bibitem{PARI}
The PARI~Group, Univ. Bordeaux.
\newblock {\em {PARI/GP} version \texttt{2.17.2}}, 2025.
\newblock Available from \url{https://pari.math.u-bordeaux.fr/}.

\bibitem{MR55380}
R.~A. Rankin.
\newblock On positive definite quadratic forms.
\newblock {\em J. London Math. Soc.}, 28:309--314, 1953.

\bibitem{MR3664816}
Maryna~S. Viazovska.
\newblock The sphere packing problem in dimension 8.
\newblock {\em Ann. of Math. (2)}, 185(3):991--1015, 2017.

\bibitem{watanabe-fundamental}
Takao Watanabe.
\newblock {\em Fundamental Domains for Arithmetic Quotients of Reductive
  Algebraic Groups}, volume~10 of {\em Lecture Note Series in Mathematics}.
\newblock Osaka Sugaku Kankokai, Osaka, 2025.
\newblock \url{https://doi.org/10.18910/102699}.

\end{thebibliography}

\end{document}